\documentclass{article}
\usepackage[utf8]{inputenc}
\usepackage{float}
\usepackage{amsmath}
\usepackage{amsthm}
\usepackage{qtree}
\usepackage{float}
\usepackage[noend,noline,fillcomment,ruled,linesnumbered]{algorithm2e}
\usepackage{enumitem}
\usepackage{hyperref}
\usepackage{breakurl}

\usepackage[backend=biber]{biblatex}
\usepackage{tgtermes}
\usepackage[scale=.9]{sourcecodepro}

\newtheorem{theorem}{Theorem}
\newtheorem{lemma}{Lemma}

\theoremstyle{remark}
\newtheorem{remark}{Remark}

\theoremstyle{definition}
\newtheorem{definition}{Definition}

\setlist[enumerate]{noitemsep,label=(\arabic*)}

\title{An algorithm for generating random mixed-arity~trees}
\author{Aleksander Kiryk}
\date{\today}

\begin{document}
\maketitle

\begin{abstract}
	Inspired by \cite{korsh} we present a new algorithm for uniformly random generation of ordered trees in which all occuring outdegrees can be specified by a given sequence of numbers. The method can be used for random generation of binary or n-ary trees, or ones with various arities. We show that the algorithm is correct and has $O(n)$ time complexity for $n$ being the desired number of nodes in the resulting tree. In the discussion part we show how some selected formulas can be derived with the use of ideas developed in the proof of correctness of the algorithm.
\end{abstract}

\section{Introduction}
Trees are data structures commonly used in mathematics and computer science. The problem of generating them in a uniformly random way has been already well studied \cite{comparison}, resulting in methods that have potential applications in areas like software testing, machine learning and statistics where unbiased sampling is usually desired. The algorithms generating random trees can be characterized by the types of constraints they accept. These constraints usually define a class of objects that we want to sample giving each of these objects an equal chance to be selected. Algorithms are known that randomly select objects from the following classes of trees:
\begin{enumerate}
	\item unordered trees restricted by an expected number of nodes \cite{comparison},
	\item ordered trees restricted by an expected number of nodes \cite{ordered},
	\item ordered trees restricted by an expected number of nodes and their degree \cite{korsh-kary},
	\item binary trees restricted by an expected number of nodes \cite{korsh} \cite{remy},
	\item and some more.
\end{enumerate}

In this article we present a method for random selection of trees from a class restricted by a sequence of all outdegrees that occur in the tree. The method is heavily inspired by a binary tree generating algorithm proposed by Korsh \cite{korsh}. Here we use a different encoding in order to represent nodes with various outdegrees. We also show that most of the desired properties of the Korsh method still hold.

It is worth noting that if the outdegrees are interpreted as arities, the algorithm can be used for generating random syntactic trees for arithmetical and logical expressions.

\section{The algorithm}
The input array $A$ of non-negative integers is expected to contain outdegrees of all the nodes that will be present in the resulting tree. As explained in the section \ref{correctness}, the constraints enforced by the contents of $A$ must be realistic, i.e. at least one tree that meets the requirements must exist. The lemma \ref{nanda} provides simple means to examine if it does.

The convention used here is that the arrays are indexed starting with $1$ and the notation $A[a..b]$ represents a subarray of $A$ starting from the element under index $a$ and ending with the element under the index $b$, both inclusive. The operator $+$ used with arrays represents concatenation.

\begin{algorithm}[H]
	\caption{\textsc{Get-Random-Tree}$(A)$}
	\SetSideCommentLeft
	\LinesNumbered
	\DontPrintSemicolon
	$A \gets \textsc{Shuffle}(A)$\;\label{shuffle}
	$n \gets 0$\;
	$k \gets 0$\;
	\For{$i \gets 1$ \KwTo $|A|$\label{search}}{%
		\textit{$n \gets n + 1 - A[i]$}\;
		\If{n = 1}{%
			$n \gets 0$\;
			$k \gets i$\;
		}
	}
	\KwRet{$A[k+1..|A|] + A[1..k]$}\;\label{rotation}
\end{algorithm}

The result will be encoded in a prefix form. For instance, for an input $A=[0,0,0,0,1,2,3]$ a potential result $[3,1,0,2,0,0,0]$ represents the following tree:

\begin{figure}[H]
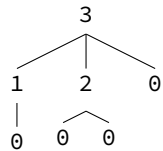

	\Tree[ .\texttt{3} [ .\texttt{1} \texttt{0} ] [ .\texttt{2} \texttt{0} \texttt{0} ] \texttt{0} ]

	\caption{The tree represented by $[3,1,0,2,0,0,0]$}
\end{figure}

The algorithm contains three components that are usually implemented as loops, which are random shuffle (line \ref{shuffle}), search for the point of rotation (line \ref{search}) and the rotation itself (line \ref{rotation}), the other parts of the algorithm can be assumed to be performed in constant time. The search and rotation are linear, procedures that perform random shuffle in linear time are also known \cite{fisher-yates}, so the overall complexity of the algorithm presented here is $O(n)$ where $n=|A|$.

\section{Correctness}\label{correctness}
\begin{definition}
	Let $C(V) := \sum_{v\in{V}}\left(1-\deg^+(v)\right)$ denote the number of elements of a set of nodes $V$ minus the sum of their outdegrees.
\end{definition}

\begin{lemma}\label{nanda}
	Let $V$ be a finite set of nodes. A tree can be constructed with a use of all elements of $V$ iff $C(V) = 1$.
\end{lemma}
\begin{proof}
	Let us use induction to prove the forward part. Let $t$ be a tree with a single node, obviously if $V$ contains only that node, then $C(V)=1$.

	Now let $t$ be a larger tree, $V$ be the set of its nodes and $v\in{V}$ be the root of $t$. We assume that the statement holds for trees smaller than $t$. We know that $C(V-\{v\})=\deg^+(v)$ because $v$ has $\deg^+(v)$ subtrees and the statement holds for all of them. Since $V$ is a union of $\{v\}$ and $V-\{v\}$ we can see that \[C(V)=C(V-\{v\})+1-\deg^+(v)=1\] and so the statement holds also for $V$.

	In the backwards part, the proof also goes by induction. If $V$ is a minimal set which meets the sufficient condition, it contains only a single leaf node. In this case $C(V)=1$, so the statement holds.

	Now we are going to show that for $|V|>1$ if the statement holds for all sets smaller than $V$, then it also holds for $V$. Let $v$ be an element of $V$ with a maximal outdegree. Since $C(V) = 1$ and leaves are the only elements of $V$ that increase the value of $C(V)$, the set $V$ must contain at least $\deg^+(v)$ leaves, otherwise $C(V)$ would not be a positive number. Let us create a tree $t$ with $v$ as the root and the $\deg^+(v)$ leaves as its leaves. Now we will treat $t$ as a single leaf node and define a set $W$ which consists of $t$ and the elements of $V$ that were not used for $t$. $W$ contains $\deg^+(v)$ less nodes than $V$ and a sum of arities lowered by $\deg^+(v)$, so $C(W)$ can be calculated as \[C(W)=C(V)-\deg^+(v)+\deg^+(v)=1\] and since $W$ is a smaller set and a tree for it can be trivially converted into a tree for $V$, the inductive step is established.
\end{proof}

\begin{remark}\label{sequence}
	Later we will use $C(V)$ in an analogous way also for sequences.
\end{remark}

\begin{definition}\label{np}
	A well-formed \emph{expression in Polish notation} (or shortly an \emph{expression}) is either a symbol representing a variable or a constant, or it is a symbol representing an operator concatenated by expressions in a number equal the arity of that operator.
\end{definition}

\begin{remark}\label{hunger}
	It is worth noting that in Polish notation the only situation in which an expression is not well-formed is when some of its operators are followed by too few subexpressions. If it is followed by too many of them, then the whole string is not well-formed, but it has a prefix which is.
\end{remark}

In this article we will use Polish notation for encoding rooted trees, where the leaves are represented by the \texttt{0} constant and the other nodes are operators with the operands being their children. We denote the nodes as digits that correspond to their outdegree.

\begin{definition}
	If a string $uw$ has a postfix $w$ of length $k$, then its \emph{$k$-rotation} is the string $wu$.
\end{definition}

\begin{lemma}\label{equality}
	For a well-formed expression of length $n$, each of its $n-1$ non-identity rotations is not well-formed.
\end{lemma}
\begin{proof}
	By definition \ref{np}, every operator in a well-formed expression must be followed by an exact number of subexpressions, so rotating a postfix of a string to its beginning has to leave at least one of the operators without some of its operands.
\end{proof}

\begin{lemma}\label{correction}
	Every string that is not a well-formed expression, but could be reordered into a well-formed expression is a rotation of a well-formed expression.
\end{lemma}
\begin{proof}
	First let us denote such a string by $uw$ and notice it always begins with a prefix $u$ containing disjoint well-formed expressions, and only the postfix $w$ that follows them is a single expression that is not well-formed. It is a direct conclusion of the remark \ref{hunger}.

	Now let us analyze $w$. First we define $1-h := C(w)$. We know that $1-h < 1$ as the operators of $w$ lack operands to become a correct expression. We could fix $w$ if we placed $h$ well-formed expressions after it. Now let us notice that $C(u)$ must equal $h$ as the assumption on the whole string $uw$ is that we can build a single well-formed expression out of its symbols meaning that $C(uw)=1$. We have previously shown that $w$ is the only expression in $uw$ that lacks operands, so it is granted that $u$ contains exactly $h$ correct expressions, therefore the $|w|$-rotation of $uw$ is a well-formed expression.
\end{proof}

\begin{theorem}
	Given a string $w$ that can become a well-formed expression representing a tree by having its characters rearranged, we can select one of these representations in a uniformly random way by first randomly selecting one of permutations of $w$, and then fixing it by an appropriate rotation.
\end{theorem}
\begin{proof}
	It follows from the lemma \ref{correction} that every incorrect expression in the set of permutations of $w$ is a rotation of a well-formed one, and lemma \ref{equality} guarantees that for every well-formed expression there is the same number of its incorrect rotations. Considering the above, randomly choosing a permutation of $w$ and fixing it, using the method from lemma \ref{correction}, guarantees uniformly random selection of a tree.
\end{proof}

\section{Discussion}
The methods used in the proof of correcntess allow us to make some side notions. For instance, every well-formed expression of length $n$ has $n-1$ incorrect rotations, and every incorrect expression is a rotation of a well-formed one. Which means the chances for a randomly reordered well-formed expression being well-formed too are $1-\frac{n-1}{n}=\frac{1}{n}$.

This, in turn, can be used in derivation of the Catalan numbers. Given a set $V$ of $n$ nodes of outdegree two, we can calculate the number $l$ of leaves that we need in order to build a proper tree using them.
\begin{align*}
	C(V) + l &= 1 \\
	n-2n + l &= 1 \\
	l &= 1 + n
\end{align*}
Now since all these nodes can be ordered in $(n+l)!=(2n+1)!$ ways, then the number of permutations that are pairwise different can be calculated by $\frac{(2n+1)!}{n!(n+1)!}$, but only $\frac{1}{(2n+1)}$ of them is correct, giving us $\frac{(2n)!}{n!(n+1)!}$ binary trees of the size $n$.

\printbibliography
\end{document}